\documentclass[prl,floatfix,showpacs,twocolumn]{revtex4}

\usepackage{times,amsmath,amsfonts,amssymb,latexsym}
\usepackage{graphicx,epsf}

\newenvironment{proof}[1][Proof]{\noindent\textbf{#1.} }{\ \rule{0.5em}{0.5em}}

\newtheorem{myproposition}{Proposition}
\newtheorem{mylemma}{Lemma}
\newtheorem{mydefinition}{Definition}
\newtheorem{mytheorem}{Theorem}

\begin{document}

\title{Vanishing quantum discord is necessary and sufficient for completely
positive maps}
\author{Alireza Shabani$^{(1)}$ and Daniel A. Lidar$^{(1,2,3)}$}
\affiliation{Departments of $^{(1)}$Electrical Engineering, $^{(2)}$Chemistry, and $%
^{(3)} $Physics\\
Center for Quantum Information Science \& Technology, University of Southern
California, Los Angeles, CA 90089, USA}

\begin{abstract}
Two long standing open problems in quantum theory are to characterize the
class of initial system-bath states for which quantum dynamics is equivalent
to (1) a map between the initial and final system states, and (2) a
completely positive (CP) map. The CP map problem is especially important,
due to the widespread use of such maps in quantum information processing and
open quantum systems theory. Here we settle both these questions by showing
that the answer to the first is ``all'', with the resulting map being
Hermitian, and that the answer to the second is that CP maps arise
exclusively from the class of separable states with vanishing quantum
discord.
\end{abstract}

\maketitle

\textit{Introduction}.---Every natural object is in contact with its
environment, so its dynamics is that of an \textquotedblleft
open\textquotedblright\ system. The problem of the formulation and
characterization of the dynamics of open systems in the quantum regime has a
long and extensive history \cite{Breuer:book}. Its importance derives from
the desire and need to be able to consistently describe the evolution of
subsystems, without having to make reference to the entire universe.
Consider a quantum system $S$ coupled to another system $B$, with respective
Hilbert spaces $\mathcal{H}_{S}$ and $\mathcal{H}_{B}$, such that together
they form one isolated system, described by the joint initial state (density
matrix) $\rho _{SB}(0)$. $B$ represent the environment, or bath, so the
object of interest is the system $S$, whose state at time $t$ is governed
according to the standard quantum-mechanical prescription by the following
quantum dynamical process (QDP): 
\begin{equation}
\rho _{S}(t)=\mathrm{Tr}_{B}[\rho _{SB}(t)]=\mathrm{Tr}_{B}[U_{SB}(t)\rho
_{SB}(0)U_{SB}(t)^{\dag }].  \label{dynamics1}
\end{equation}
The propagator $U_{SB}(t)$ is a unitary operator, the solution to the
Schrodinger equation $\dot{U}_{SB}=-(i/\hbar )[H_{SB},U_{SB}]$, where $%
H_{SB} $ is the joint system-bath Hamiltonian. $\mathrm{Tr}_{B}$ represents
the partial trace operation, corresponding to an averaging over the bath
degrees of freedom \cite{Breuer:book}.

The QDP (\ref{dynamics1}) is a transformation from $\rho _{SB}(0)$ to $\rho
_{S}(t)$. However, since we are not interested in the state of the bath, it
is natural to ask: under which conditions is the QDP a map from $\rho
_{S}(0) $ to $\rho _{S}(t)$ \cite{CP-com1}? When is this map linear? When is
it completely positive (CP) \cite{Kraus:83}? These are fundamental questions
which have been the subject of intense studies with a long history \cite%
{Sudarshan:61,Pechukas:94,Alicki:95,Lindblad:96,Stelmachovic:01,Jordan:04,Zyczkowski:04,Carteret:05,Rodriguez:07}%
, also more recently in the context of non-Markovian master equations \cite%
{budini:053815breuer:022103vacchini:022112}. One reason that these questions
have attracted so much interest is the fundamental role played by CP maps in
quantum information \cite{Nielsen:book} and open quantum systems theory \cite%
{Breuer:book}. CP maps are the ``workhorse'' in these fields, and hence an
understanding of their domain of validity is essential. For this reason it
is perhaps surprising that the problem of identifying the general physical
conditions under which CP maps are valid has remained open since it was
first posed in a vigorous debate \cite{Pechukas:94,Alicki:95}. In
particular, while \emph{sufficient} conditions have been developed for
complete positivity \cite{Alicki:95,Rodriguez:07}, it is not known which is
the most general class of states for which the QDP (\ref{dynamics1}) is
always CP, for arbitrary $U_{SB}$. In this work we settle this old open
question. We prove that the QDP\ yields a CP map $\rho _{S}(0)\mapsto \rho
_{S}(t)$ iff $\rho _{SB}(0)$ has vanishing \textquotedblleft quantum
discord\textquotedblright\ \cite{Ollivier:01}, i.e., is purely classicaly
correlated.

In order to arrive at this result we introduce a class of states we call
\textquotedblleft special-linear\textquotedblright\ (SL), with the property
of being of full measure in the set of mixed bipartite states. We show that
the QDP\ (\ref{dynamics1}) is always a linear Hermitian map $\Phi_{\mathrm{H}%
}:\rho _{S}(0)\mapsto \rho _{S}(t)$ if $\rho _{SB}(0)$ in the SL-class.
Vanishing discord states are a subset of SL states, and CP maps are a subset
of Hermitian maps; we use the SL construction to prove our main result about
CP maps. We then argue that the restriction to the SL-class can be lifted,
and that in fact the QDP (\ref{dynamics1}) is \emph{always} a linear
Hermitian map, for arbitrary $\rho _{SB}(0)$. This result settles another
old open question: is quantum subsystem dynamics always a map, and if so, of
what kind?

\textit{Linear maps}.--- A linear map is ``Hermitian'' if it preserves the
Hermiticity of its domain. We first present an operator sum representation
for arbitrary and Hermitian linear maps:

\begin{mytheorem}
\label{th1}A map $\Phi :\mathfrak{M}_{n}\mapsto \mathfrak{M}_{m}$ (where $%
\mathfrak{M}_{n}$ is the space of $n\times n$ matrices) is linear iff it can
be represented as 
\begin{equation}
\Phi (\rho )\text{ }=\sum_{\alpha }E_{\alpha }\rho E_{\alpha }^{\prime
\dagger }  \label{linear}
\end{equation}%
where the \textquotedblleft left and right operation
elements\textquotedblright\ $\{E_{\alpha }\}$ and $\{E_{\alpha }^{\prime }\}$
are, respectively, $m\times n$ and $n\times m$ matrices.\newline
$\Phi_{\mathrm{H}}$ is a Hermitian map iff 
\begin{equation}
\Phi_{\mathrm{H}} (\rho )=\sum_{\alpha }c_{\alpha }E_{\alpha }\rho E_{\alpha
}^{\dagger },\quad c_{\alpha }\in \mathbb{R}.  \label{eq:QM}
\end{equation}
\end{mytheorem}

(See Refs.~\cite{Hill:73,shabani-2006} for a proof). A linear map is called
\textquotedblleft completely positive\textquotedblright\ (CP) if it is a
Hermitian map with $c_{\alpha }\geq 0$ $\forall \alpha $. It turns out that
there is a tight connection between CP and Hermitian maps \cite%
{Jordan:04,Carteret:05}: a map is Hermitian iff it can be written as the
difference of two CP maps.

The definition of a CP map $\Phi _{\mathrm{CP}}$ implies that it can be
expressed in the Kraus operator sum representation \cite{Kraus:83}: $\rho
_{S}(t)=\sum_{\alpha }E_{\alpha }(t)\rho _{S}(0)E_{\alpha }^{\dagger
}(t)=\Phi _{\mathrm{CP}}(t)[\rho _{S}(0)]$. If the operation elements $%
E_{\alpha }$ satisfy ${\sum_{\alpha }}E_{\alpha }^{\dagger }E_{\alpha }=I$
then $\mathrm{Tr}[\rho _{S}(t)]=1$. The standard argument in favor of the
ubiquitousness of CP maps is that, since $S $ may be coupled with $B$, the
maps $\Phi _{\mathrm{Ph}}$ describing physical processes on $S$ should be
such that all their extensions into higher dimensional spaces should remain
positive, i.e., $\Phi _{\mathrm{Ph}}\otimes I_{n}\geq 0$ $\forall n\in 
\mathbb{Z}^{+}$, where $I_{n}$ is the $n$-dimensional identity operator;
this means that $\Phi _{\mathrm{Ph}}$ is a CP map \cite{Choi:75}. However,
one may question whether this is the right criterion for describing quantum
dynamics on the grounds that this imposes restrictions on the allowed class
of initial system-bath states \cite{Pechukas:94}. An alternative viewpoint
is to seek a description that applies to \emph{arbitrary} $\rho _{SB}(0)$.
However, it was recently shown \cite{Hayashi:03} that the QDP (\ref%
{dynamics1}) with \emph{arbitrary} $\rho _{SB}(0)$ becomes a CP map iff a
most restrictive condition is satisfied by $U_{SB}(t)$, namely, it must be
locally unitary: $U_{SB}(t)=U_{S}(t)\otimes U_{B}(t)$, i.e., the effective
system-bath interaction must vanish. If one gives up the consistency
condition $\rho _{S}=\mathrm{Tr}_{B}[\rho _{SB}] $ for all $\rho _{S}$, or
gives up linearity except in the weak coupling regime, CP maps arise for
more general initial states \cite{Alicki:95}.

A recent breakthrough due to Rodriguez \textit{et al}. \cite{Rodriguez:07}
shows that CP maps arise for arbitrary $U_{SB}$ even for certain
non-factorized initial conditions, namely provided the initial state $\rho
_{SB}(0)$ is invariant under the application of a complete set of orthogonal
one-dimensional projections on $S$, i.e., the state has vanishing quantum
discord. Here we show that vanishing quantum discord is not only sufficient
but also necessary for the QDP to induce a CP map (Theorem~\ref{VQD}). We go
further and ask whether the larger class of Hermitian maps is compatible
with general initial conditions. We shall show that this is indeed the case
(Theorem~\ref{H-States}).

\textit{Special-linear states}.--- We now define a class of states we call
\textquotedblleft special-linear\textquotedblright\ (SL) states for which
the QDP (\ref{dynamics1}) always results in a linear, Hermitian map. An
arbitrary bipartite state on $\mathcal{H}_{S}\otimes \mathcal{H}_{B}$ can be
written as 
\begin{equation}
\rho _{SB}=\sum_{ij}\varrho _{ij}|i\rangle \langle j|\otimes \phi _{ij},
\label{eq:rho_SB}
\end{equation}%
where $\{|i\rangle \}_{i=1}^{\dim \mathcal{H}_{S}}$ is an orthonormal basis
for $\mathcal{H}_{S}$, and $\{\phi _{ij}\}_{i,j=1}^{\dim \mathcal{H}_{S}}:%
\mathcal{H}_{B}\mapsto \mathcal{H}_{B}$ are normalized such that if $\mathrm{%
Tr}[\phi _{ij}]\neq 0$ then $\mathrm{Tr}[\phi _{ij}]=1$. The corresponding
reduced system and bath states are then $\rho _{S}=\sum_{(i,j)\in \mathcal{C}%
}\varrho _{ij}|i\rangle \langle j|$, where $\mathcal{C}\equiv \{(i,j)|%
\mathrm{Tr}[\phi _{ij}]=1\}$, and $\rho _{B}(0)=\sum_{i}\varrho _{ii}\phi
_{ii}$. Hermiticity and normalization of $\rho _{SB}$, $\rho _{S}$, and $%
\rho _{B}$ imply $\varrho _{ij}=\varrho _{ji}^{\ast }$, $\phi _{ij}=\phi
_{ji}^{\dag }$, and $\sum_{i}\varrho _{ii}=1$.

\begin{mydefinition}
\label{def:SL}A bipartite state $\rho _{SB}$ is in the SL-class\ iff either $%
\mathrm{Tr}[\phi _{ij}]=1$ or $\phi _{ij}=0,$ $\forall i,j$.
\end{mydefinition}

The following is a key result which we prove at the end:

\begin{mytheorem}
\label{H-States}If $\rho _{SB}(0)$ is an SL-class state then the QDP (\ref%
{dynamics1}) is a linear, Hermitian map $\Phi _{\mathrm{H}}:\rho
_{S}(0)\mapsto \rho _{S}(t)$.
\end{mytheorem}

Next we need to be precise about the block structure associated with a
matrix $A=[a_{ij}]$:

\begin{mydefinition}
We call two diagonal elements $a_{i_{1}i_{1}}$ and $a_{i_{B}i_{B}}$
\textquotedblleft block-connected via the path $\{i_{b}\}_{b=2}^{B-1}$%
\textquotedblright\ if there exists a set of unequal indexes $%
\{i_{b}\}_{b=1}^{B}$ such that $\{a_{i_{b}i_{b+1}}\}_{b=1}^{B-1}$ are all
non-zero, i.e., they can be connected via a path that involves only
horizontal and vertical (but not diagonal) moves. The \textquotedblleft
block-index set\textquotedblright\ $\mathcal{D}_{A}^{(\alpha )}$ is the set
of all index pairs $\{(i,j)\}$ of the elements of the $\alpha $th block of $%
A $.
\end{mydefinition}

This is just the standard notion of a block in a matrix, possibly before
permutation matrices are applied to sort it into the standard block-diagonal
structure. We are now ready to state our main result.

\begin{mylemma}
\label{th:CP}Let $\rho _{SB}(0)$ [Eq.~(\ref{eq:rho_SB})] be an SL-class
state, let $\Phi \equiv \lbrack \phi _{ij}]=\bigoplus_{\alpha }\Phi
^{(\alpha )}$ (a supermatrix), and let $\{\Pi _{\alpha }\equiv
\sum_{(i,i)\in \mathcal{D}_{\Phi }^{(\alpha )}}|i\rangle \langle
i|\}_{\alpha }$ be a complete set of projectors from $\mathcal{H}_{S}$ to $%
\mathcal{H}_{S}$. Let $\mathcal{C}_{\Phi }^{(\alpha )}\equiv \{(i,j)\in 
\mathcal{D}_{\Phi }^{(\alpha )}|\mathrm{Tr}[\phi _{ij}]=1\}$ and 
\begin{equation}
\rho _{S}^{(\alpha )}\equiv \Pi _{\alpha }\rho _{S}(0)\Pi _{\alpha
}/p_{\alpha }=\sum_{(i,j)\in \mathcal{C}_{\Phi }^{(\alpha )}}\varrho
_{ij}|i\rangle \langle j|/p_{\alpha },  \label{eq:PrP}
\end{equation}%
where $p_{\alpha }=\mathrm{Tr}[\rho _{S}(0)\Pi _{\alpha }]$. Let $\rho
_{B}^{(\alpha )}$ be a density matrix. The Hermitian map $\Phi _{\mathrm{H}%
}:\rho _{S}(0)\mapsto \rho _{S}(t)$ induced by the QDP (\ref{dynamics1}) is
a CP map iff $(\Phi ^{(\alpha )})_{ij}=\{0$ $\mathrm{or}$ $\rho
_{B}^{(\alpha )}\}$ $\forall (i,j)\in \mathcal{D}_{\Phi }^{(\alpha )}$:%
\begin{equation}
\rho _{SB}(0)=\sum_{\alpha }p_{\alpha }\rho _{S}^{(\alpha )}\otimes \rho
_{B}^{(\alpha )}.  \label{eq:rSB}
\end{equation}
\end{mylemma}

Clearly, $\rho _{S}^{(\alpha )}$ can be thought of as the post-measurement
state arising with probability $p_{\alpha }$ from $\rho _{S}(0)$ after the
application of the projective measurement described by the set $\{\Pi
_{\alpha }\}$. Moreover, $\rho _{SB}(0)$ is not merely separable:

\begin{mytheorem}
\label{VQD}The Hermitian map $\Phi _{\mathrm{H}}:\rho _{S}(0)\mapsto \rho
_{S}(t)$ is a CP map iff the initial system-bath state $\rho _{SB}(0)$\ has
vanishing quantum discord (VQD), i.e., can be written as:%
\begin{equation}
\rho _{SB}(0)=\sum_{k,\alpha }\Pi _{\alpha }^{k}\rho _{SB}(0)\Pi _{\alpha
}^{k},  \label{eq:rSB-P}
\end{equation}%
where $\{\Pi _{\alpha }^{k}\}$\ are one-dimensional projectors onto the
eigenvectors of $\rho _{S}^{(\alpha )}$, and $\sum_{k}\Pi _{\alpha }^{k}=\Pi
_{\alpha }$.
\end{mytheorem}

\begin{proof}
By expanding $\rho _{S}^{(\alpha )}$ as $\sum_{k}p_{\alpha }^{k}\Pi _{\alpha
}^{k}$, with $p_{\alpha }^{k}=\mathrm{Tr}[\rho _{S}(0)\Pi _{\alpha
}^{k}]\geq 0$ and $\sum_{k}p_{\alpha }^{k}=1$, we obtain using Eq. (\ref%
{eq:rSB}): $\rho _{SB}(0)=\sum_{\alpha }\rho _{S}^{(\alpha )}\otimes \rho
_{B}^{(\alpha )}=\sum_{k,\alpha }p_{\alpha }^{k}\Pi _{\alpha }^{k}\otimes
\rho _{B}^{(\alpha )}$, which implies Eq. (\ref{eq:rSB-P}). On the other
hand $\sum_{k,\alpha }\Pi _{\alpha }^{k}\rho _{SB}(0)\Pi _{\alpha }^{k}$ is
the state after a non-selective projective measurement $\{\Pi _{\alpha
}^{k}\}$ on $S$, so that $\rho _{SB}(0)=\sum_{k,\alpha }p_{\alpha }^{k}\Pi
_{\alpha }^{k}\otimes \rho _{B}^{(\alpha )}$.
\end{proof}

The quantum discord has a deep information-theoretic origin and
interpretation, for the details of which we refer the reader to Ref.~\cite%
{Ollivier:01}; we shall merely remark that when the discord vanishes all the
information about $B$ that exists in the $S$-$B$ correlations is locally
recoverable just from the state of $S$, which is not the case for a general
separable state of $S$ and $B$. In this sense a VQD state is
\textquotedblleft completely classical\textquotedblright .

\begin{proof}[Proof of Lemma \protect\ref{th:CP}]
We start with necessity; sufficiency will turn out to be trivial. Let us
assume that the Hermitian map $\Phi _{\mathrm{H}}:\rho _{S}(0)\mapsto \rho
_{S}(t)$ induced by the QDP, $\rho _{S}(t)=\mathrm{Tr}_{B}[U\rho
_{SB}(0)U^{\dag }],$\ is CP, and determine the class of allowed initial
states. We start from an SL-class state since we know (Theorem \ref{H-States}%
) that in this case the QDP (\ref{dynamics1}) is indeed equivalent to a
Hermitian map. Let $\widetilde{M}=|\Psi \rangle \langle \Psi |$, where $%
|\Psi \rangle =\frac{1}{\sqrt{d_{S}}}\sum_{i=1}^{d_{S}}|i\rangle \otimes
|i\rangle $ is a maximally entangled state over $\mathcal{H}_{S}\otimes 
\mathcal{H}_{S}$, and where $d_{S}=\dim \mathcal{H}_{S}$. It follows
directly from Eq.~(\ref{Linear})\ below that $\Phi _{\mathrm{H}}[|i\rangle
\langle j|]=\mathrm{Tr}_{B}[U|i\rangle \langle j|\otimes \phi
_{ij}U^{\dagger }]$. Thus the Choi matrix \cite{Choi:75} for $\Phi _{\mathrm{%
L}}$ is%
\begin{align}
\mathcal{M}& \equiv (\mathcal{I}\otimes \Phi _{\mathrm{H}})[\widetilde{M}]=%
\frac{1}{d_{S}}\sum_{ij}|i\rangle \langle j|\otimes \Phi _{\mathrm{H}%
}[|i\rangle \langle j|]  \notag \\
& =\frac{1}{d_{S}}\sum_{ij}|i\rangle \langle j|\otimes \mathrm{Tr}%
_{B}[U|i\rangle \langle j|\otimes \phi _{ij}U^{\dagger }],
\end{align}%
We assume that $\mathcal{M}$ is positive as this is equivalent to $\Phi _{%
\mathrm{H}}$ being CP \cite{Choi:75}. A useful fact is that a matrix $A\ $is
positive iff every principal submatrix of $A$ is positive (a principal
submatrix is the matrix obtained by deleting from $A$ some number of columns
and rows with equal indexes). Therefore, let us focus on the pair of rows
and columns $(k,l)$ $(k\neq l)$ of $d_{S}\mathcal{M}$, and consider the $%
2\times 2$ principal submatrix 
\begin{equation}
P_{kl}=\left( 
\begin{array}{cc}
\mathrm{Tr}_{B}[U|k\rangle \langle k|\otimes \phi _{kk}U^{\dagger }] & 
\mathrm{Tr}_{B}[U|k\rangle \langle l|\otimes \phi _{kl}U^{\dagger }] \\ 
\mathrm{Tr}_{B}[U|l\rangle \langle k|\otimes \phi _{lk}U^{\dagger }] & 
\mathrm{Tr}_{B}[U|l\rangle \langle l|\otimes \phi _{ll}U^{\dagger }]%
\end{array}%
\right) .
\end{equation}%
The submatrix $P_{kl}$ must be positive for any $U$, and we choose to
examine the case $U=\frac{1}{\sqrt{2}}(I\otimes I-iX\otimes A)$, where $A$
is Hermitian and unitary (hence $A^{2}=I$), and $X=|k\rangle \langle
l|+|l\rangle \langle k|+\sum_{i\neq k,l}|i\rangle \langle i|$. This will
allow us to find restrictions on $\{\phi _{kl}\}$. Note that it follows from
Hermiticity of $A$, $\phi _{kk}$ and $\phi _{ll}$, and from $\phi
_{kl}^{\dag }=\phi _{lk}$, that $\mathrm{Tr}[A\phi _{kk}],\mathrm{Tr}[A\phi
_{ll}]\in \mathbb{R}$, and that $\mathrm{Tr}[A\phi _{kl}]=(\mathrm{Tr}[A\phi
_{lk}])^{\ast }$. Thus some algebra yields: 
\begin{eqnarray}
P_{kl} &=&\frac{1}{4}\left( 
\begin{array}{cccc}
t_{kk} & ia & ib & t_{kl} \\ 
-ia & t_{kk} & t_{kl} & -ib \\ 
-ib^{\ast } & t_{kl} & t_{ll} & -ic \\ 
t_{kl} & ib^{\ast } & ic & t_{ll}%
\end{array}%
\right) ,  \label{eq:Pkl} \\
a &=&\mathrm{Tr}[A\phi _{kk}]\in \mathbb{R},\quad b=\mathrm{Tr}[A\phi _{kl}],
\notag \\
c &=&\mathrm{Tr}[A\phi _{ll}]\in \mathbb{R},\quad t_{ij}=\mathrm{Tr}[\phi
_{ij}]=1\text{ or }0.  \notag
\end{eqnarray}%
To proceed we require the following Lemma (proof at end):

\begin{mylemma}
\label{lem:A}If $\mathrm{Tr}[AX]=0$ for any unitary and Hermitian matrix $A$
then $X=0$.
\end{mylemma}

\begin{myproposition}
\label{prop0}If $\phi _{kk}=0$ or $\phi _{ll}=0$ then $\phi _{kl}=\phi
_{lk}=0$.
\end{myproposition}

\begin{proof}
Assume that $\phi _{ll}=0$ or $\phi _{kk}=0$, but not both, so that either $%
(t_{ll}=0,t_{kk}=1)$, or $(t_{ll}=1,t_{kk}=0)$. Construct the principal
submatrix obtained by deleting rows and columns $1$ and $3$ from $P_{kl}$.
This leaves a principal submatrix with eigenvalues $(1\pm \sqrt{1+4|b|^{2}}%
)/8$. The positivity of these requires $b=\mathrm{Tr}[A\phi _{kl}]=0$, so
that by Lemma \ref{lem:A} $\phi _{kl}=\phi _{lk}^{\dag }=0$. When $\phi
_{ll}=\phi _{kk}=0$ the same principal submatrix$\allowbreak $ has
eigenvalues $\pm |b|$, so that again $\phi _{kl}=\phi _{lk}^{\dag }=0$.
\end{proof}

\begin{myproposition}
\label{prop1}If all of $\phi _{kk},\phi _{ll},\phi _{kl}\neq 0$ then $\phi
_{kk}=\phi _{ll}=\phi _{kl}=\phi _{lk}$.
\end{myproposition}

\begin{proof}
After a couple of elementary row and column operations on $P_{kl}$ we obtain:%
\begin{equation}
P_{kl}^{\prime }=\left( 
\begin{array}{cc}
1_{2} & B \\ 
B^{\dag } & 1_{2}%
\end{array}%
\right) ;\quad 1_{2}=\left( 
\begin{array}{cc}
1 & 1 \\ 
1 & 1%
\end{array}%
\right) ,\quad B=\left( 
\begin{array}{cc}
ib & ia \\ 
ic & ib^{\ast }%
\end{array}%
\right) .
\end{equation}%
Diagonalizing the two diagonal blocks $1_{2}$ using $Q=\frac{1}{\sqrt{2}}%
(I+i\sigma _{y})$ yields $P_{kl}^{\prime \prime }=Q^{\oplus 2}P_{kl}^{\prime
}(Q^{\dag })^{\oplus 2}$, where%
\begin{align*}
P_{kl}^{\prime \prime }& =\left( 
\begin{array}{cc}
C & D \\ 
D^{\dag } & C%
\end{array}%
\right) ;\quad C=\left( 
\begin{array}{cc}
2 & 0 \\ 
0 & 0%
\end{array}%
\right) ,\quad D=i\left( 
\begin{array}{cc}
\alpha & \beta \\ 
\gamma & \delta%
\end{array}%
\right) ; \\
\alpha & =(a+b+b^{\ast }+c)/2,\quad \beta =(a-b+b^{\ast }-c)/2, \\
\gamma & =(-a-b+b^{\ast }+c)/2,\quad \delta =(-a+b+b^{\ast }-c)/2.
\end{align*}%
Positivity of $P_{kl}$ implies that also $P_{kl}^{\prime \prime }>0$, so
that we can again apply the principal submatrix method. Let $e(i,j)$ denote
the eigenvalues of the $P_{kl}^{\prime \prime }$ submatrix obtained by
retaining only the $i$th and $j$th rows and columns of $P_{kl}^{\prime
\prime }$. We find $e(1,4)=1\pm \sqrt{1+|\beta |^{2}}$, $e(2,3)=1\pm \sqrt{%
1+|\gamma |^{2}}$ and $e(2,4)=\pm |\alpha |^{2}$. Since all these
eigenvalues must be positive we conclude that $\alpha =\beta =\delta =0$,
i.e., $\mathrm{Tr}[A\phi _{kk}]=\mathrm{Tr}[A\phi _{kl}]=\mathrm{Tr}[A\phi
_{lk}]=\mathrm{Tr}[A\phi _{ll}]$. Applying Lemma \ref{lem:A} we have $%
\mathrm{Tr}[A(\phi _{kk}-\phi _{kl})]=0$, so that $\phi _{kk}=\phi _{kl}$,
and similarly $\phi _{kl}=\phi _{lk}=\phi _{ll}$.
\end{proof}

It is simple to check that the only permissible case not covered by
Propositions \ref{prop0} and \ref{prop1} is when $\phi _{kk},\phi _{ll}\neq
0 $ and $\phi _{kl}=\phi _{lk}=0$; in this case we have no further
restrictions.

\begin{mylemma}
\label{lem:Phi}The matrix $\Phi \equiv \lbrack \phi _{ij}]$ can be
decomposed as $\Phi =\bigoplus_{\alpha }\Phi ^{(\alpha )}$, where $(\Phi
^{(\alpha )})_{(i,j)\in \mathcal{D}_{\Phi }^{(\alpha )}}=\phi ^{(\alpha )}$
(a constant) or $0$.
\end{mylemma}

\begin{proof}
Every matrix is a direct sum of blocks (possibly only one). Therefore our
task is to prove that the matrix elements of the $\alpha $th block $\Phi
^{(\alpha )}$ obey $(\Phi ^{(\alpha )})_{(i,j)\in \mathcal{D}_{\Phi
}^{(\alpha )}}=\phi ^{(\alpha )}$ or $0$. Collecting the results above we
see that there are only four cases: Proposition \ref{prop0} $\Longrightarrow 
$ (i) $\phi _{kk}=\phi _{kl}=\phi _{lk}=\phi _{ll}=0$, (ii) $\phi _{kk}=\phi
_{kl}=\phi _{lk}=0$ and $\phi _{ll}\neq 0$; Proposition \ref{prop1} $%
\Longrightarrow $ (iii) $\phi _{kk}=\phi _{kl}=\phi _{lk}=\phi _{ll}\neq 0$;
(iv) $\phi _{kk},\phi _{ll}\neq 0$ and $\phi _{kl}=\phi _{lk}=0$. First note
that if $\phi _{kk}=0$ then by cases (i)\ and (ii) also $\phi _{kl}=\phi
_{lk}=0$ $\forall l$, i.e., the row and column crossing at a zero diagonal
element must be zero. Now let $\Psi _{ij}^{(\alpha )}$ denote the $2\times 2$
principal submatrix $\{\Phi _{ii}^{(\alpha )},\Phi _{ij}^{(\alpha )};\Phi
_{ji}^{(\alpha )},\Phi _{jj}^{(\alpha )}\}$, $i\neq j$. Assume $\Phi
_{ij}^{(\alpha )}\neq 0$ and consider $\Psi _{ij}^{(\alpha )}$. Only case
(iii) applies, so $\Phi _{ii}^{(\alpha )}=\Phi _{ij}^{(\alpha )}=\Phi
_{ji}^{(\alpha )}=\Phi _{jj}^{(\alpha )}$. We can use this to show that any
two block-connected diagonal elements are equal. Indeed, assume that $\Phi
_{i_{1}i_{1}}^{(\alpha )}$ and $\Phi _{i_{B}i_{B}}^{(\alpha )}$ are both
non-zero and block-connected via the path $\{i_{b}\}_{b=2}^{B-1}$. Then by
case (iii) all elements of each member of the set of principal submatrices $%
\{\Psi _{i_{b}i_{b+1}}^{(\alpha )}\}_{b=1}^{B-1}$ are equal, and since
successive members always share a diagonal element, their elements are all
equal, to an element we call $\phi ^{(\alpha )}$. We have thus shown that $%
(\Phi ^{(\alpha )})_{(i,j)\in \mathcal{D}_{\Phi }^{(\alpha )}}=\phi
^{(\alpha )}$ or $0$. Finally, note that case (iv) with $\phi _{kk}\neq \phi
_{ll}$ can only arise between two different blocks, since if $\phi _{kk}\neq
\phi _{ll}$ the previous argument shows that they cannot be block-connected.
\end{proof}

We are now ready to conclude the proof of Lemma \ref{th:CP}: It follows from
Lemma \ref{lem:Phi} that $(\Phi )_{ij}=(\Phi ^{(\alpha )})_{ij}=\phi
^{(\alpha )}$ or $0$ for $(i,j)\in \mathcal{D}_{\Phi }^{(\alpha )}$.
Moreover, since $\rho _{SB}(0)$ is an SL-class state, $\mathrm{Tr}[\phi
^{(\alpha )}]=1$. Thus the total index set $\mathcal{D}_{\Phi }$ for the
initial state $\rho _{SB}(0)$ splits into a union of disjoint index sets $%
\mathcal{D}_{\Phi }^{(\alpha )}$, so that Eqs.~(\ref{eq:PrP}) and (\ref%
{eq:rSB}) are satisfied, where $\rho _{S}(0)=\sum_{\alpha }\sum_{(i,j)\in 
\mathcal{C}_{\Phi }^{(\alpha )}}\varrho _{ij}|i\rangle \langle
j|=\sum_{\alpha }\Pi _{\alpha }\rho _{S}(0)\Pi _{\alpha }=\sum_{\alpha
}\sigma _{S}^{(\alpha )}$, where $\sigma _{S}^{(\alpha )}=p_{\alpha }\rho
_{S}^{(\alpha )}$ and where $\rho _{B}^{(\alpha )}\equiv \phi ^{(\alpha )}$.
Here $\Pi _{\alpha }$ is the projector onto the subspace corresponding to
block $\alpha $ (as defined above). Next we need to show that the $\rho
_{B}^{\alpha }$'s are density matrices. From the properties of the $\phi
^{(\alpha )}$ we already have $\mathrm{Tr}\rho _{B}^{(\alpha )}=1$, so what
is left to prove is that $\rho _{B}^{(\alpha )}>0$. Indeed, by definition of
positivity $\langle i^{(\alpha )}|\langle \psi _{B}|\rho _{SB}(0)|i^{(\alpha
)}\rangle |\psi _{B}\rangle >0$ for any state $|i^{(\alpha )}\rangle $ in
the support of $\Pi _{\alpha }$ and any bath state $|\psi _{B}\rangle $.
Inserting $\rho _{SB}(0)=\sum_{\alpha }p_{\alpha }\rho _{S}^{(\alpha
)}\otimes \rho _{B}^{(\alpha )}$ into this inequality, we find $\langle \psi
_{B}|\rho _{B}^{(\alpha )}|\psi _{B}\rangle >0,\forall |\psi _{B}\rangle \in 
\mathcal{H}_{B}$. This completes the proof of necessity. Sufficiency: using
the spectral decomposition $\rho _{B}^{\alpha }=\sum_{j}\lambda _{j}^{\alpha
}|\lambda _{j}^{\alpha }\rangle \langle \lambda _{j}^{\alpha }|$ and
defining $E_{ij}^{\alpha }\equiv \langle \beta _{i}|U|\lambda _{j}^{\alpha
}\rangle \Pi _{\alpha }:\mathcal{H}_{S}\mapsto \mathcal{H}_{S}$, where $%
\{|\beta _{i}\rangle \}$ is an orthonormal basis for $\mathcal{H}_{B}$, we
have, using Eqs.~(\ref{dynamics1}) and (\ref{eq:PrP}): 
\begin{equation}
\rho _{S}(t)=\mathrm{Tr}_{B}[U\rho _{SB}(0)U^{\dagger }]=\sum_{\alpha
ij}\lambda _{i}^{\alpha }E_{ij}^{\alpha }\rho _{S}(0)E_{ij}^{\alpha \dagger
}.  \label{Linear-project}
\end{equation}%
Now we simply note that if $\rho _{SB}(0)$ satisfies Eq.~(\ref{eq:rSB}) with 
$\rho _{B}^{(\alpha )}>0$ (i.e., $\lambda _{j}^{\alpha }>0$), then Eq.~(\ref%
{Linear-project}) is already in the form of a CP map, with operation
elements $\{\sqrt{\lambda _{i}^{\alpha }}E_{ij}^{\alpha }\}_{\alpha ij}$.
\end{proof}

\textit{Discussion}.--- What is the physical meaning of fixing the bath-only
operators $\phi _{ij}$, as is required in our formulation? The answer is
that this corresponds to fixing the initial system-bath correlations:\ the
purely classical part is determined by the $\phi _{ii}$, while the quantum
part is determined by the $\phi _{ij}$ with $i\neq j$. Further, note that $%
\mathrm{Tr}[|j\rangle \langle i|\otimes I_{B}\rho _{SB}]=\varrho _{ij}%
\mathrm{Tr}[\phi _{ij}]$, so that non-SLness can also be written as $\mathrm{%
Tr}[|j\rangle \langle i|\otimes I_{B}\rho _{SB}]=0$, i.e., as $\langle
|i\rangle \langle j|\otimes I_{B},\rho _{SB}\rangle =0$ (Hilbert-Schmidt
inner product $\langle A,B\rangle \equiv \mathrm{Tr}[A^{\dag }B]$) and hence 
$\rho _{SB}$ must lie in the hyperplane orthogonal to $|i\rangle \langle
j|\otimes I_{B}$. Thus non-SL-class states are confined to a
lower-dimensional surface in the space of bipartite states, and must be 
\emph{sparse}. Note that, conversely, the SL condition $\mathrm{Tr}[\phi
_{ij}]=1$ yields $\mathrm{Tr}[|j\rangle \langle i|\otimes I_{B}\rho
_{SB}]=\varrho _{ij}$, which is not a constraint since $\varrho _{ij}$ is
arbitrary. Moreover, using a mapping from affine to linear maps \cite%
{Jordan:05}, it is not hard to show that the zero-measure subset of non-SL
states does not spoil Theorem \ref{H-States}, i.e., the QDP\ (\ref{dynamics1}%
) is a linear, Hermitian map from $\rho _{S}(0)\mapsto \rho _{S}(t)$ for 
\emph{any} initial state $\rho _{SB}(0)$.

\textit{Conclusions}.---In this work we have identified the conditions for
the validity of quantum subsystem dynamics. In particular, we have found the
precise initial state conditions for the ubiquitous class of CP maps. This
establishes a foundation for their widespread use in quantum information and
open systems theory. We have also shown that the basic quantum mechanical
transformation (\ref{dynamics1}) is always representable as a Hermitian map
between the initial and final system states. This result establishes that
quantum subsystem dynamics is always a meaningful concept.

\textit{Proofs}.--- In order to prove Theorem \ref{H-States} we first need:

\begin{mylemma}
\label{pr:QP=HM}If $\rho _{SB}(0)$ is an SL-class state then the QDP (\ref%
{dynamics1}) is a linear map $\Phi_\mathrm{L}: \rho _{S}(0)\mapsto\rho
_{S}(t)$.
\end{mylemma}

\begin{proof}
Consider the singular value decomposition (SVD) $\phi _{ij}=\sum_{\alpha
}\lambda _{\alpha }^{ij}|x_{ij}^{\alpha }\rangle \langle y_{ij}^{\alpha }|$,
where $\lambda _{\alpha }^{ij}$ are the singular values and $|x_{ij}^{\alpha
}\rangle $ ($\langle y_{ij}^{\alpha }|$) are the right (left) singular
vectors. Let $\{|\psi _{k}\rangle \}$ be an orthonormal basis for the bath
Hilbert space ${\mathcal{H}}_{B}$, and define the system operators $%
V_{kij}^{\alpha }\equiv \langle \psi _{k}|U_{SB}|x_{ij}^{\alpha }\rangle $, $%
W_{kij}^{\alpha }\equiv \langle \psi _{k}|U_{SB}|y_{ij}^{\alpha }\rangle $.
Since $\rho _{SB}(0)$ is an SL-class state, a QDP (\ref{dynamics1})
generated by an arbitrary unitary evolution $U_{SB}$ yields [recall $%
\mathcal{C}\equiv \{(i,j)|\mathrm{Tr}[\phi _{ij}]=1\}$]: 
\begin{align}
\rho _{S}(t)& =\mathrm{Tr}_{B}[\rho _{SB}(t)]=\sum_{ij}\varrho _{ij}\mathrm{%
Tr}_{B}[U_{SB}|i\rangle \langle j|\otimes \phi _{ij}U_{SB}^{\dag }]  \notag
\\
& =\sum_{(i,j)\in \mathcal{C};k,\alpha }\lambda _{\alpha }^{ij}\varrho
_{ij}V_{kij}^{\alpha }|i\rangle \langle j|(W_{kij}^{\alpha })^{\dag }.
\label{eq:1}
\end{align}%
Now note that $P_{i}\rho _{S}(0)P_{j}=\varrho _{ij}|i\rangle \langle j|$,
where $P_{i}\equiv |i\rangle \langle i|$ is a projector and $(i,j)\in 
\mathcal{C}$. Therefore:%
\begin{eqnarray}
\Phi _{\mathrm{L}}[\rho _{S}(0)] &\equiv &\sum_{(i,j)\in \mathcal{C}%
;k,\alpha }\lambda _{\alpha }^{ij}V_{kij}^{\alpha }P_{i}\rho
_{S}(0)P_{j}(W_{kij}^{\alpha })^{\dag }  \label{eq:p} \\
&=&\sum_{(i,j)\in \mathcal{C};k,\alpha }\lambda _{\alpha }^{ij}\varrho
_{ij}V_{kij}^{\alpha }|i\rangle \langle j|(W_{kij}^{\alpha })^{\dag },
\label{Linear}
\end{eqnarray}%
which equals $\rho _{S}(t)$ according to Eq.~(\ref{eq:1}). This defines the
linear map $\Phi _{\mathrm{L}}=\{E_{ijk\alpha },E_{ijk\alpha }^{\prime }\}$,
whose left and right operation elements are $\{E_{ijk\alpha }\equiv \sqrt{%
\lambda _{\alpha }^{ij}}V_{kij}^{\alpha }P_{i}\}$ and $\{E_{ijk\alpha
}^{\prime }\equiv \sqrt{\lambda _{\alpha }^{ij}}W_{kij}^{\alpha }P_{j}\}$,
respectively.
\end{proof}

We need to show that $\rho _{S}(t)=\Phi _{\mathrm{H}}[\rho _{S}(0)]=\rho
_{S}^{\dag }(t)$ if $\rho _{S}(0)=\rho _{S}^{\dag }(0)$. This is now a
simple calculation which uses Eqs.~(\ref{eq:p}) and (\ref{Linear}), the
definitions of $V_{kij}^{\alpha }$ and $W_{kij}^{\alpha }$, $\phi _{ij}=\phi
_{ji}^{\dag }$, and the SVD of $\phi _{ij}$.

\begin{proof}[Proof of Theorem \protect\ref{H-States}]
We need to show that $\rho _{S}(t)=\Phi _{\mathrm{H}}[\rho _{S}(0)]=\rho
_{S}^{\dag }(t)$ if $\rho _{S}(0)=\rho _{S}^{\dag }(0)$. This is now a
simple calculation which uses Eqs.~(\ref{eq:p}) and (\ref{Linear}), the
definitions of $V_{kij}^{\alpha }$ and $W_{kij}^{\alpha }$, $\phi _{ij}=\phi
_{ji}^{\dag }$, and the SVD of $\phi _{ij}$.
\end{proof}

\begin{proof}[Proof of Lemma \protect\ref{lem:A}]
Since $A$ is unitary and Hermitian its eigenvalues are both roots of unity
and real, i.e., it can always be parameterized in the form $A=UDU^{\dagger }$%
, where $U$ is unitary and the diagonal matrix $D$ has diagonal elements $%
\pm 1$. Consider two special choices of $D$: $D_{1}=\mathrm{diag}%
(+1,+1,...,+1)=I$ and $D_{2}=\mathrm{diag}(-1,+1,+1,...,+1)=I-2|0\rangle
\langle 0|$. Since $\mathrm{Tr}[D_{1}U^{\dagger }XU]=\mathrm{Tr}%
[D_{2}U^{\dagger }XU]=0$ we find $\mathrm{Tr}[(D_{1}-D_{2})U^{\dagger }XU]=0$%
, or $\mathrm{Tr}[|0\rangle \langle 0|U^{\dagger }XU]=0$. However, $U$ is
arbitrary, so that $\langle \psi |X|\psi \rangle =0$, $\forall |\psi \rangle 
$ ($|\psi \rangle =U|0\rangle $). This can only be true if $X=0$.
\end{proof}

\textit{Acknowledgements}.--- This work was funded by NSF Grants No.
CCF-0523675 and CCF-0726439 (to D.A.L).

\end{document}